\documentclass[pra,preprintnumbers,amsmath,amssymb,superscriptaddress]{revtex4}

\usepackage{bm}
\usepackage{graphicx}
\usepackage{amsbsy}
\usepackage{amsmath}
\usepackage{amsfonts}
\usepackage{amsthm}
\usepackage{mathrsfs}

\begin{document}

\theoremstyle{plain}
\newtheorem{theorem}{Theorem}
\newtheorem{lemma}[theorem]{Lemma}
\newtheorem{corollary}[theorem]{Corollary}
\newtheorem{conjecture}[theorem]{Conjecture}
\newtheorem{proposition}[theorem]{Proposition}
\newcommand{\PT}{\mathrm{PTL}}
 \newcommand{\C}{\mathbb{C}}
  \newcommand{\F}{\mathbb{F}}
  \newcommand{\N}{\mathbb{N}}
  \renewcommand{\P}{\mathbb{P}}
  \newcommand{\R}{\mathbb{R}}
  \newcommand{\Z}{\mathbf{Z}}
  \renewcommand{\a}{\mathbf{a}}
  \renewcommand{\b}{\mathbf{b}}
  \renewcommand{\i}{\mathbf{i}}
  \renewcommand{\j}{\mathbf{j}}
  \renewcommand{\c}{\mathbf{c}}
  \newcommand{\e}{\mathbf{e}}
  \newcommand{\f}{\mathbf{f}}
  \newcommand{\g}{\mathbf{g}}
  \newcommand{\gl}{\mathbf{GL}}
  \newcommand{\m}{\mathbf{m}}
  \newcommand{\n}{\mathbf{n}}
  \newcommand{\bNP}{\mathbf{NP}}
  \newcommand{\bNPC}{\mathbf{NPC}}
  \newcommand{\p}{\mathbf{p}}
  \newcommand{\bP}{\mathbb{P}}
  \newcommand{\bPo}{\mathbf{Po}}
  \newcommand{\q}{\mathbf{q}}
  \newcommand{\s}{\mathbf{s}}
  \newcommand{\bt}{\mathbf{t}}
  \newcommand{\T}{\mathbf{T}}
  \newcommand{\U}{\mathbf{U}}
  \renewcommand{\u}{\mathbf{u}}
  \renewcommand{\v}{\mathbf{v}}
  \newcommand{\V}{\mathbf{V}}
  \newcommand{\w}{\mathbf{w}}
  \newcommand{\W}{\mathbf{W}}
  \newcommand{\x}{\mathbf{x}}
  \newcommand{\X}{\mathbf{X}}
  \newcommand{\y}{\mathbf{y}}
  \newcommand{\Y}{\mathbf{Y}}
  \newcommand{\z}{\mathbf{z}}
  \newcommand{\0}{\mathbf{0}}
  \newcommand{\1}{\mathbf{1}}
  \newcommand{\Gam}{\mathbf{\Gamma}}
  \newcommand{\bGamma}{\Gam}
  \newcommand{\Lam}{\mathbf{\Lambda}}
  \newcommand{\lam}{\mbox{\boldmath{$\lambda$}}}
  \newcommand{\bA}{\mathbf{A}}
  \newcommand{\bB}{\mathbf{B}}
  \newcommand{\bC}{\mathbf{C}}
  \newcommand{\bH}{\mathbf{H}}
  \newcommand{\bL}{\mathbf{L}}
  \newcommand{\bM}{\mathbf{M}}
  \newcommand{\bc}{\mathbf{c}}
  \newcommand{\cA}{\mathcal{A}}
  \newcommand{\cB}{\mathcal{B}}
  \newcommand{\cC}{\mathcal{C}}
  \newcommand{\cD}{\mathcal{D}}
  \newcommand{\cE}{\mathcal{E}}
  \newcommand{\cF}{\mathcal{F}}
  \newcommand{\cG}{\mathcal{G}}
  \newcommand{\cH}{\mathcal{H}}
  \newcommand{\cI}{\mathcal{I}}
  \newcommand{\cK}{\mathcal{K}}
  \newcommand{\cL}{\mathcal{L}}
  \newcommand{\cM}{\mathcal{M}}
  \newcommand{\cO}{\mathcal{O}}
  \newcommand{\cP}{\mathcal{P}}
  \newcommand{\cR}{\mathcal{R}}
  \newcommand{\cS}{\mathcal{S}}
  \newcommand{\cT}{\mathcal{T}}
  \newcommand{\cU}{\mathcal{U}}
  \newcommand{\cV}{\mathcal{V}}
  \newcommand{\cW}{\mathcal{W}}
  \newcommand{\cX}{\mathcal{X}}
  \newcommand{\cY}{\mathcal{Y}}
  \newcommand{\cZ}{\mathcal{Z}}
  \newcommand{\rE}{\mathrm{E}}
  \newcommand{\rH}{\mathrm{H}}
  \newcommand{\rU}{\mathrm{U}}
  \newcommand{\Cp}{\mathrm{Cap\;}}
  \newcommand{\lan}{\langle}
  \newcommand{\ran}{\rangle}
  \newcommand{\an}[1]{\lan#1\ran}
  \def\diag{\mathop{{\rm diag}}\nolimits}
  \newcommand{\hs}{\hspace*{\parindent}}
  \newcommand{\cl}{\mathop{\mathrm{Cl}}\nolimits}
  \newcommand{\tr}{\mathop{\mathrm{Tr}}\nolimits}
  \newcommand{\Aut}{\mathop{\mathrm{Aut}}\nolimits}
  \newcommand{\argmax}{\mathop{\mathrm{arg\,max}}}
  \newcommand{\Eig}{\mathop{\mathrm{Eig}}\nolimits}
  \newcommand{\Gr}{\mathop{\mathrm{Gr}}\nolimits}
  \newcommand{\Fr}{\mathop{\mathrm{Fr}}\nolimits}
  \newcommand{\trans}{^\top}
  \newcommand{\opt}{\mathop{\mathrm{opt}}\nolimits}
  \newcommand{\per}{\mathop{\mathrm{perm}}\nolimits}
  \newcommand{\haff}{\mathrm{haf\;}}
  \newcommand{\perio}{\mathrm{per}}
  \newcommand{\conv}{\mathrm{conv\;}}
  \newcommand{\Cov}{\mathrm{Cov}}
  \newcommand{\inter}{\mathrm{int}}
  \newcommand{\dist}{\mathrm{dist}}
  \newcommand{\inn}{\mathrm{in}}
  \newcommand{\grank}{\mathrm{grank}}
  \newcommand{\mrank}{\mathrm{mrank}}
  \newcommand{\krank}{\mathrm{krank}}
  \newcommand{\out}{\mathrm{out}}
  \newcommand{\orient}{\mathrm{orient}}
  \newcommand{\Pu}{\mathrm{Pu}}
  \newcommand{\rdc}{\mathrm{rdc}}
  \newcommand{\range}{\mathrm{range\;}}
  \newcommand{\Sing}{\mathrm{Sing\;}}
  \newcommand{\topo}{\mathrm{top}}
  \newcommand{\undir}{\mathrm{undir}}
  \newcommand{\Var}{\mathrm{Var}}
  \newcommand{\rC}{\mathrm{C}}
  \newcommand{\rF}{\mathrm{F}}
  \newcommand{\rL}{\mathrm{L}}
  \newcommand{\rM}{\mathrm{M}}
  \newcommand{\rO}{\mathrm{O}}
  \newcommand{\rR}{\mathrm{R}}
  \newcommand{\rS}{\mathrm{S}}
  \newcommand{\rT}{\mathrm{T}}
  \newcommand{\pr}{\mathrm{pr}}
  \newcommand{\inte}{\mathrm{int}}
  \newcommand{\inv}{\mathrm{inv}}
  \newcommand{\pers}{\per_s}
  \newcommand{\del}{\boldsymbol{\delta}}
  \renewcommand{\alph}{\boldsymbol{\alpha}}
  \newcommand{\bet}{\boldsymbol{\beta}}
  \newcommand{\gam}{\boldsymbol{\gamma}}
  \newcommand{\sig}{\boldsymbol{\sigma}}
  \newcommand{\zet}{\boldsymbol{\zeta}}
  \newcommand{\et}{\boldsymbol{\eta}}
  \newcommand{\xit}{\boldsymbol{\xi}}
  \newcommand{\perm}{\mathrm{perm\;}}
  \newcommand{\adj}{\mathrm{adj\;}}
  \newcommand{\rank}{\mathrm{rank\;}}
  \newcommand{\set}[1]{\{#1\}}
  \newcommand{\spec}{\mathrm{spec\;}}
  \newcommand{\supp}{\mathrm{supp\;}}
  \newcommand{\Tr}{\mathrm{Tr\;}}
  \newcommand{\vol}{\text{vol}}
  \newcommand{\be}{\begin{equation}}
  \newcommand{\ee}{\end{equation}}
   \newcommand{\nn}{\nonumber\\}

\theoremstyle{definition}
\newtheorem{definition}{Definition}

\theoremstyle{remark}
\newtheorem*{remark}{Remark}
\newtheorem{example}{Example}

\title{The minimum Renyi entropy output of a quantum channel is locally additive}

\author{Gilad Gour}\thanks{Supported by NSERC}\email{gour@ucalgary.ca}
\affiliation{Institute for Quantum Science and Technology, and
Department of Mathematics and Statistics,
University of Calgary, 2500 University Drive NW,
Calgary, Alberta, Canada T2N 1N4}

\author{Todd Kemp}\thanks{Supported by NSF CAREER Award DMS-1254807}\email{tkemp@math.ucsd.edu}
\affiliation{Department of Mathematics, 
University of California, San Diego
9500 Gilman Drive,
La Jolla, CA  92093-0112}

\begin{abstract}
We show that the minimum Renyi entropy output of a quantum channel is locally additive for Renyi parameter $\alpha>1$. While our work extends the results of~\cite{GF} (in which local additivity was proven for $\alpha=1$), 
it is based on several new techniques that incorporate the \emph{multiplicative} nature of $\ell_p$-norms, in contrast to the \emph{additivity} property of the von-Neumann entropy.
Our results demonstrate that the counterexamples to the Renyi additivity conjectures exhibit purely global effects of quantum channels.  
Interestingly, the approach presented here can not be extended to Renyi entropies with parameter $\alpha<1$.
\end{abstract}

\maketitle

\section{Introduction}

One of the most fundamental questions in quantum information theory concerns the ability to send information over a noisy quantum communication channel~\cite{Holevo06,c1,c2,c3,c4,c5,c6,c7,c8,c9,Has09,Yard08,Brandao,Fuk10,Bra11}. 
Unlike classical channels, quantum channels exhibits an unintuitive phenomenon in which the optimal rate to transmit reliably classical or quantum information is not additive under taking tensor products of two (or more) quantum channels~\cite{Has09,Yard08}. The question of whether the product state classical capacity (i.e. Holevo capacity) is additive or not was an open problem for more than a decade and was shown by Shor~\cite{c6} to be equivalent to three other additivity conjectures; namely, the additivity of entanglement of formation, the strong super-additivity of entanglement of formation, and the additivity of the minimum entropy output of a quantum channel. 

The discovery that all these quantities are not additive~\cite{Has09} left with it key problems in the field wide open. One such problem is how much entanglement between input signal states is needed to violate additivity. A partial answer for that was given in recent work~\cite{GF,FGA}, where it was shown that the minimum entropy output of a quantum channel is locally additive. This result indicated that while entanglement is needed to violate additivity, arbitrarily small amount of entanglement will not be sufficient. Here we extend this result by showing that the minimum Renyi output entropies with parameter $\alpha$ greater than 1 are also locally additive. However, unlike the von-Neumann entropy ($\alpha=1$), for $\alpha>1$ the Renyi entropies are given in terms of the log of the $\alpha$-norms (also known as $\ell_p$-norms, where for notational convenience we rename $p=\alpha$ presently) which are multiplicative under tensor product. Therefore, in order to prove local multiplicativity of the output $\alpha$-norms of a quantum channel, it is not possible to use exactly the same techniques as those used in the case $\alpha=1$ since the latter relied heavily on the additive nature of the von-Neumann entropy. 

This paper is organized as follows. In Section~\ref{preli} we provide the precise definition of local additivity, including notations and preliminaries that will be used in the subsequent sections. Section~\ref{sec:main} is devoted to the statement and proof of the main result. The proof of the main result is based on 3 substantial lemmas that will be proved in Sections~\ref{sec:der},~\ref{sec:dir}, and~\ref{sec:add}. Finally, in Section~\ref{sec:concluding} we end with a few concluding remarks.

\section{Notations and Preliminaries}\label{preli}

Quantum channels are described in terms of completely-positive trace preserving linear maps (CPT maps).
A CPT map $\mathcal{N}: H_{d_{\rm in}}\to H_{d_{\rm out}}$ takes the set of $d_{\rm in}\times d_{\rm in}$ Hermitian matrices 
$H_{d_{\rm in}}$ to a subset of the set of all $d_{\rm out}\times d_{\rm out}$ Hermitian matrices
$H_{d_{\rm out}}$. Any finite dimensional quantum channel can be characterized in terms of a unitary embedding followed by a partial trace (the Stinespring dilation theorem): for any CPT map $\mathcal{N}$ there exists an ancillary space of Hermitian matrices $H_{E}$ such that
$$
\mathcal{N}(\rho)=\tr_{E}\left[U(\rho\otimes |0\rangle_{E}\langle 0|) U^{\dag}\right]
$$
where $\rho\in H_{d_{\rm in}}$ and $U$ is a unitary matrix mapping states $|\psi\rangle|0\rangle_E$ with
$|\psi\rangle\in H_{d_{\text{in}}}$ to 
$H_{d_{\text{out}}}\otimes H_{E}$.
 
For $\alpha\ge0$, the minimum $\alpha$-Renyi entropy output of a quantum channel $\mathcal{N}$ is defined by
\be\label{minim}
S_{\alpha}^{\min}(\mathcal{N})\equiv\min_{\rho\in H_{d_{\text{in}},+,1}}S_\alpha\left(\mathcal{N}(\rho)\right)\;,
\ee
where $H_{d_{\text{in}},+,1}\subset H_{d_{\text{in}}}$ is the set of all $d_{\rm in}\times d_{\rm in}$ positive semi-definite matrices with trace $=1$ (i.e. density matrices), and $$S_\alpha(\rho)\equiv\frac{1}{1-\alpha}\log\tr(\rho^\alpha)$$ is the $\alpha$-Renyi entropy with $0\leq \alpha\leq \infty$, where for $\alpha=0,1,\infty$ the Renyi entropies are defined in terms of the limits. For $0\leq \alpha\leq 1$ the Renyi entropies are concave in $\rho$, and therefore it follows that the minimization can be taken over all rank one matrices $\rho=|\psi\rangle\langle\psi|$ in $H_{d_{\text{in}},+,1}$. While for $\alpha>1$ the Renyi entropy is not concave in general (only Schur concave),  we can still take the minimization over rank-1 matrices since the $\alpha$-norm $\|\rho\|_\alpha=\left[\tr(\rho^\alpha)\right]^{1/\alpha}$ is convex for $\alpha\geq 1$.  To see why, note that for $\rho=\sum_jp_j|\psi_j\rangle\langle\psi_j|$, we have $\|\mathcal{N}(\rho)\|_\alpha\leq\sum_jp_j \|\mathcal{N}(|\psi_j\rangle\langle\psi_j|)\|_\alpha$, so that
\begin{align*}
S_{\alpha}\left(\mathcal{N}(\rho)\right)& =\frac{\alpha}{1-\alpha}\log\|\mathcal{N}(\rho)\|_\alpha
\geq \frac{\alpha}{1-\alpha}\log\sum_jp_j \|\mathcal{N}(|\psi_j\rangle\langle\psi_j|)\|_\alpha\\
& \geq\frac{\alpha}{1-\alpha}\log\max_{j}\|\mathcal{N}(|\psi_j\rangle\langle\psi_j|)\|_\alpha
=\min_jS_\alpha\left(\mathcal{N}(|\psi_j\rangle\langle\psi_j|)\right)\;.
\end{align*}
Therefore, for all $0\leq\alpha\leq\infty$ the minimum in~\eqref{minim} can be taken over all rank 1 matrices in $H_{d_{\text{in}},+,1}$.

For any such rank one density matrix $\rho=|\psi\rangle\langle\psi|$ we can define a bipartite pure state
$|\Psi\rangle=U|\psi\rangle|0\rangle_{E}$ in the bipartite subspace $\mathcal{K}\equiv \{|\Psi\rangle\big|\;|\psi\rangle\in H_{d_{\rm in}}\}$. Therefore, the minimum
Renyi entropy output of the channel $\mathcal{N}$ can be expressed in terms of the Renyi entanglement of the bipartite subspace  $\mathcal{K}$ defined by 
$$
E_\alpha(\mathcal{K})\equiv\min_{|\phi\rangle\in\mathcal{K}\;,\;\|\phi\|=1}E_\alpha(|\phi\rangle)\;,
$$
where $E_\alpha(|\phi\rangle)\equiv S_\alpha\left(\tr_{E}(|\phi\rangle\langle\phi|)\right)$ is the Renyi entropy of entanglement. In~\cite{GN} it was pointed out that $E_\alpha(\mathcal{K})=0$ unless $\dim\mathcal{K}\leq (d_{\rm out}-1)(\dim H_{E}-1)$. This claim follows directly from the fact that the number of (bipartite) states in an unextendible product basis is at least $d_{\rm out}+\dim H_{E}-1$, cf.\ \cite{Ben99}.

With these notations, the non-additivity of the minimum Renyi entropy output of a quantum channel is equivalent to 
the existence of two subspaces $\mathcal{K}_1\subset\mathbb{C}^{n_1}\otimes\mathbb{C}^{m_1}$ and $\mathcal{K}_2\subset\mathbb{C}^{n_2}\otimes\mathbb{C}^{m_2}$ such that
$$
E_\alpha(\mathcal{K}_1\otimes\mathcal{K}_1)<E_\alpha(\mathcal{K}_1)+E_\alpha(\mathcal{K}_2)\;.
$$

\subsection*{Local Minimum/Maximum}\label{local}

Let $\mathcal{K}\subset\mathbb{C}^{n}\otimes\mathbb{C}^{m}$ be a subspace of bipartite entangled states.
Since the bipartite Hilbert space $\mathbb{C}^{n}\otimes\mathbb{C}^{m}$ is isomorphic to the Hilbert space
of all $n\times m$ complex matrices $\mathbb{C}^{n\times m}$, we can view any bipartite state 
$|\psi\rangle^{AB}=\sum_{i,j}x_{ij}|i\rangle|j\rangle$ in $\mathcal{K}$ as an $n\times m$ matrix $x$.
The reduced density matrix of $|\psi\rangle^{AB}$ is then given by $\rho_r\equiv\tr_{B}|\psi\rangle^{AB}\langle\psi|=xx^*$,
and the $\alpha$-Renyi entropy of entanglement of $|\psi\rangle^{AB}$ is given by
\begin{equation}\label{entropy}
E_\alpha(x)\equiv\frac{1}{1-\alpha}\log\tr[(xx^*)^\alpha]\;.
\end{equation}
In our notations, instead of using a dagger, we use $x^*$ to denote the hermitian conjugate of the matrix $x$.

Since the $\log$ function is continuous and monotonic, instead of showing that $E_\alpha$ is locally additive for $\alpha>1$, we will show that 
$$
Q_\alpha(x)=\tr[(xx^*)^\alpha]
$$
is locally multiplicative.

If $x\in\mathcal{K}$ is a local minimum of $E_\alpha$ in $\mathcal{K}$ (i.e. $x\in\mathcal{K}$ is a local \emph{maximum} of $Q_\alpha$ in $\mathcal{K}$), then there exists a neighbourhood of $x$ in $\mathcal{K}$
such that $x$ is the minimum in that neighbourhood. Any state in a neighbourhood of $x$ can be written as
$ax+by$, where $a,b\in\mathbb{C}$ and $y\in\mathcal{K}$ is a matrix orthogonal to $x$; i.e. $\tr(xy^*)=0$.
We also assume that the state is normalized so that $|a|^2+|b|^2=1$.
Now, since the function $E_\alpha(x)$ (or $Q_\alpha(x)$) is independent of a global phase, we can assume that $a$ is a positive \emph{real} number.
We can also assume that $b$ is real since we can absorb its phase into $y$ (adding a phase to $y$ will not change its orthogonality to $x$). Thus, any normalized state in a neighbourhood of $x$ can be written as
$$
\frac{x+ty}{\sqrt{1+t^2}}\;\;\text{with}\;\;\tr(xy^*)=0\;,
$$
where $t\equiv b/a$ is a small real number and $y$ is normalized (i.e. $\tr(yy^*)=1$). 

\begin{definition}\label{maindef}
$\;$\\
\textbf{(a)} A matrix $x\in\mathcal{K}$ is said to be a critical point of $Q_\alpha(x)$ in $\mathcal{K}$ if
$$
D_{y}Q_\alpha(x)\equiv \frac{d}{dt}Q_\alpha\left(\frac{x+ty}{\sqrt{1+t^2}}\right)\Big|_{t=0}=0\;\;\;\forall\;y\in x^{\perp}
$$
where the notation $D_{y}Q_\alpha(x)$ indicates that we are taking the directional derivative of $Q_\alpha$ in the direction of $y$,
and $x^\perp\subset\mathcal{K}$ denotes the subspace of all the matrices $y$ in $\mathcal{K}$ for which
$\tr(xy^*)=0$.\\
\textbf{(b)}  A matrix $x\in\mathcal{K}$ is said to be a non-degenerate local maximum of $Q_\alpha(x)$ in $\mathcal{K}$ if it is critical and
$$
D_{y}^{2}Q_\alpha(x)\equiv \frac{d^2}{dt^2}Q_\alpha\left(\frac{x+ty}{\sqrt{1+t^2}}\right)\Big|_{t=0}<0\;\;\;\forall\;y\in x^{\perp}.
$$ 
(This is a maximum for $Q_\alpha$, which gives a minimum for $E_\alpha$, since $\frac{1}{1-\alpha}<0$.) Moreover, a critical $x\in\mathcal{K}$ is said to be degenerate if there exists at least one direction $y$ such that
$D_{y}^{2}Q_\alpha(x)=0$.
\end{definition}

To be clear: {\em local additivity} of $E_\alpha$ is the statement that if $x^A$ and $x^B$ are local minima for $E_\alpha$ in two subspaces $\mathcal{K}^A$ and $\mathcal{K}^B$, then $x^A\otimes x^B$ is a local minimum for $E_\alpha$ in $\mathcal{K}^A\otimes\mathcal{K}^B$.  We make this completely precise (also considering a subtle point of degenerate local minima) in Theorem \ref{thm.main} below.

In our calculations we will assume that $x$ is diagonal (or equivalently, the bipartite state $x$ that represents is given in its 
Schmidt form). This assumption results in no loss of generality, due to 
the singular value decomposition theorem; namely, we can always find unitary matrices $u\in\mathbb{C}^{n\times n}$ and $v\in\mathbb{C}^{m\times m}$ such that $uxv$ is an $n\times m$ diagonal matrix with non-negative real numbers (the singular values of $x$) on the diagonal. Since $E_\alpha(x)=E_\alpha(uxv)$, we can assume without loss of generality that $x$ is a diagonal matrix.  

\section{Main Results}\label{sec:main}

In this section we state and prove the main result of this paper.
The proof is based on 3 lemmas that will be proved in 3 subsequent sections.

\begin{theorem} \label{thm.main}
Let $\mathcal{K}^{A}$ and  $\mathcal{K}^{B}$ be two subspaces of $n_1\times m_1$ and $n_2\times m_2$ complex matrices, respectively. Let $x^A$ and $x^B$ be two normalized complex matrices in $\mathcal{K}^A$ and $\mathcal{K}^B$, respectively. Then, for $\alpha>1$:\\

{\rm \bf{(a)}} If $x^A$ and $x^B$ are non-degenerate local minima of $E_\alpha$ in $\mathcal{K}^{A}$ and  $\mathcal{K}^{B}$, respectively, then $x^A\otimes x^B$ is a non-degenerate local minimum of $E_\alpha$ in $\mathcal{K}^{A}\otimes\mathcal{K}^{B}$.\\

{\rm \bf{(b)}} If $x^A$ and $x^B$ are local minima of $E_\alpha$ in $\mathcal{K}^{A}$ and  $\mathcal{K}^{B}$, with at least one of them being non-degenerate, then $x^A\otimes x^B$ is a local minimum of $E_\alpha$ in $\mathcal{K}^{A}\otimes\mathcal{K}^{B}$.
\end{theorem}

Implicit in the above theorem is the fact that if $x^A$ and $x^B$ are critical points of $E_\alpha$ in $\mathcal{K}^{A}$
and $\mathcal{K}^{B}$, respectively, then $x^A\otimes x^B$ is a critical point of $E_\alpha$
in $\mathcal{K}^{A}\otimes\mathcal{K}^{B}$. This fact was observed in~\cite{AIM08} (see also~\cite{Maxim}), 
and was later stated in~\cite{FGA}.
It follows from the linearity in $y$ of the condition given in Eq.~(\ref{critical}) (see the next section) for critical points. We will therefore focus in this section on the higher order directional derivatives of $E_\alpha$ (or equivalently of $Q_\alpha$).

For the proof of Theorem \ref{thm.main}, we can assume without loss of generality that $n_1=m_1$, $n_2=m_2$, by padding the matrices with extra rows/columns of $0$s.
From the singular valued decomposition (see the argument below Definition~\ref{maindef}) we can assume 
without loss of generality that $x^A=\diag\{\sqrt{p_1},\ldots,\sqrt{p_{n_1}}\} $ 
and $x^B=\diag\{\sqrt{q_1},\ldots,\sqrt{q_{n_2}}\}$, 
where $p_j$ and $q_k$ are non-negative and $\sum_{j=1}^{n_1}p_i=\sum_{k=1}^{n_2}q_{j}=1$.

We first assume that both $x^A$ and $x^B$ are non-degenerate local maxima of $Q_\alpha$.
We need to show that $D_{y}^{2}Q_\alpha(x)<0$ for all $y\in x^\perp$, where $x\equiv x^A\otimes x^B$.
The most general
$y\in \left(x^A\otimes x^B\right)^\perp$ can be written as
\begin{equation}\label{generalform}
y=c_1x^A\otimes y^{B}+c_2y^{A}\otimes x^B+c_3y'\;,
\end{equation}
where $y^{A}\in (x^A)^\perp$, $y^{B}\in (x^B)^\perp$, and $y'\in\left(x^A\right)^\perp\otimes \left(x^B\right)^\perp$
are all normalized. The numbers $c_j$ can be chosen to be real because we can absorb their phases into $y^{A}$, $y^{B}$, and $y'$. They also
satisfy $c_{1}^{2}+c_{2}^{2}+c_{3}^{2}=1$, so that $y$ is normalized.

\begin{lemma}\label{directions}  If $x^A$ and $x^B$ are critical points, then
\be
D_{y}^{2}Q_\alpha(x)=c_{1}^{2}D_{x^A\otimes y^{B}}^{2}Q_\alpha(x)+c_{2}^{2}D_{y^{A}\otimes x^B}^{2}Q_\alpha(x)+c_{3}^{2}D_{y'}^{2}Q_\alpha(x).
\ee
\end{lemma}

It is therefore enough to consider the three directions $x^A\otimes y^{B}$, $y^{A}\otimes x^B$ and $y'$ separately. We will show that in each of this directions, the second order derivatives $D_{x^A\otimes y^{B}}^{2}Q_\alpha(x)$, $D_{y^{A}\otimes x^B}^{2}Q_\alpha(x)$, and $D_{y'}^{2}Q_\alpha(x)$ are all negative, so that $D_{y}^{2}Q_\alpha(x)<0$.

Consider first the simple case where $y=x^A\otimes y^{B}$. Here we have
\begin{align}\label{simple}
Q_\alpha\left(\frac{x+ty}{\sqrt{1+t^2}}\right)& =Q_\alpha\left(x^A\otimes \frac{x^B+ty^{B}}{\sqrt{1+t^2}}\right)
=Q_\alpha\left(x^A\right)Q_\alpha\left(\frac{x^B+ty^{B}}{\sqrt{1+t^2}}\right)\;.
\end{align}
Since $x^B$ is a non-degenerate local maximum, and $Q_\alpha(x^A)>0$, we must have $D_{y}^{2}Q_\alpha(x)<0$ for $y=x^A\otimes y^{B}$. The case $y=y^{A}\otimes x^B$ is similar.
 
Consider now the case in which $y\in \left(x^A\right)^\perp\otimes \left(x^B\right)^\perp$. 
To prove that $D_{y}^{2}Q_\alpha(x)<0$ in this case, we will use the following explicit computation.
 
\begin{lemma}\label{derivatives}
Let $x\in\mathcal{K}\subset\mathbb{C}^{n\times n}$ and $y\in x^{\perp}$ both normalized.
Denote the eigenvalues of $\rho\equiv xx^*$ by $\{p_j\}_{j=1}^{n}$, and decompose the complex matrix
$y=w+iz$ such that the $n\times n$ matrices $w$ and $z$ are both Hermitian. Then
\be\label{deriv}
\frac{1}{2\alpha}D_{y}^{2}Q_\alpha(x)=-\tr\left[\rho^\alpha\right]+\tr\left[w\Phi_{\rho}^{-}(w)+z\Phi_{\rho}^{+}(z)\right]\;,
\ee 
where $\Phi^{\pm}_{\rho}$ are self-adjoint linear maps acting on the Hilbert space of $n\times n$ complex matrices (equipped with the Hilbert-Schmidt inner product), defined by the following Hadamard product:
\be \label{e.phi.simple} 
\left[\Phi^{\pm}_{\rho}(y)\right]_{jk} \equiv \phi_{jk}^{\pm}\,y_{jk}\;\;\;\;\text{where}\;\;\;\;\phi_{jk}^\pm \equiv \frac{p_j^{\alpha-1/2}\pm p_k^{\alpha-1/2}}{p_j^{1/2}\pm p_k^{1/2}}. 
\ee
Furthermore, if $p_j=p_k$ then $\phi_{jk}^{+}= p_{j}^{\alpha-1}$ and $\phi_{jk}^{-}=(2\alpha-1)p_{j}^{\alpha-1}$.
\end{lemma}
Note that the second order derivative is well behaved even if some $p_j$ are zero.
Now, denote by $\rho^{A}\equiv x^Ax^{A*}$ and $\rho^{B}\equiv x^Bx^{B*}$ the two density matrices associated with the two local maxima. From the lemma above it follows that 
$D_{y}^{2}Q_\alpha(x^A\otimes x^B)<0$ for all $y\in \left(x^A\right)^\perp\otimes \left(x^B\right)^\perp$ if and only if
\be\label{main}
\tr\left[w\Phi_{\rho^A\otimes\rho^B}^{-}(w)+z\Phi_{\rho^A\otimes\rho^B}^{+}(z)\right]< \Tr[(\rho^A)^\alpha]\Tr[(\rho^B)^\alpha] \;\;\;\;,\;\;\;\forall\;y=w+iz\in\left(x^A\right)^\perp\otimes \left(x^B\right)^\perp\;.
\ee
Since $x^{A}$ and $x^B$ are local maxima of $Q_\alpha$ in their respective subspaces, it follows that  both $D_{y^{A}}^{2}Q_\alpha(x^{A})<0$ and $D_{iy^{A}}^{2}Q_\alpha(x^{A})<0$, since both $y^{A}$ and $iy^{A}$ belong to $\left(x^{A}\right)^\perp$. In particular,
\begin{align}\label{complex}
0 &>\frac{1}{4\alpha}\left(D_{y^{A}}^{2}Q_\alpha(x^A)+D_{iy^{A}}^{2}Q_\alpha(x^A)\right)\nonumber\\
 &=-\Tr[(\rho^A)^\alpha]+\frac{1}{2}\tr\left[w^{A}\left(\Phi_{\rho^A}^{-}+\Phi_{\rho^A}^{+}\right)(w^{A})+z^{A}\left(\Phi_{\rho^A}^{-}+\Phi_{\rho^A}^{+}\right)(z^{A})\right]\nonumber\\
 &=-\Tr[(\rho^A)^\alpha]+\tr\left[y^{A*}\left(\Phi_{\rho^A}^{-}+\Phi_{\rho^A}^{+}\right)(y^{A})\right]
\end{align}
where the last equality follows from the self-adjointness of $\Phi_{\rho^A}^{-}+\Phi_{\rho^A}^{+}$ with respect to the Hilbert-Schmidt inner product, and the decomposition $y^{A}=w^{A}+iz^{A}$ where $w^A$ and $z^A$ are Hermitian. We therefore arrive at the following inequalities:
\begin{equation} \label{assumption}
\begin{aligned}
\tr\left[y^{A*}\left(\Phi_{\rho^A}^{-}+\Phi_{\rho^A}^{+}\right)(y^{A})\right]&<\Tr[(\rho^A)^\alpha]\;\;\;\;\;\forall\;y^A\in\left(x^{A}\right)^\perp\\
\tr\left[y^{B*}\left(\Phi_{\rho^B}^{-}+\Phi_{\rho^B}^{+}\right)(y^{B})\right]&<\Tr[(\rho^B)^\alpha]\;\;\;\;\;\forall\;y^B\in\left(x^{B}\right)^\perp.
\end{aligned}
\end{equation}

In the final step towards the proof of the theorem, we will be using the following key operator estimate.

\begin{lemma} \label{additivity}
Let $\Psi$ be the self-adjoint linear operator acting on $\cK^A\otimes\cK^B$, given by the following Hadamard product:
\be
[\Psi(y)]_{jk,\ell m} = \psi_{jk,\ell m}\,y_{jk,\ell m}
\ee
where
\be
\psi_{jk,\ell m}=(\phi_{j\ell}^{+A}+\phi_{j\ell}^{-A})(\phi_{km}^{+B}+\phi_{km}^{-B})=\frac{p_j^\alpha-p_\ell^\alpha}{p_j-p_\ell}\frac{q_k^\alpha-q_m^\alpha}{q_k-q_m}
\ee
and $\{p_j\}$ and $\{q_k\}$ are the eigenvalues of $\rho^A$ and $\rho^B$.
Then
\be\label{thm}
\Phi_{\rho^{A}\otimes\rho^{B}}^{\pm}\leq \Psi\;.
\ee
\end{lemma}

With this lemma at hand, we are ready to prove inequality~\eqref{main}, and thus Theorem \ref{thm.main}(a).
First, observe that
\be\label{beg}
\tr\left[w\Phi_{\rho^A\otimes\rho^B}^{-}(w)+z\Phi_{\rho^A\otimes\rho^B}^{+}(z)\right]
\leq\tr\left[w\Psi(w)+z\Psi(z)\right]=\tr\left[y^*\Psi(y)\right].
\ee
Now, note that any
$y\in \left(x^A\right)^\perp\otimes \left(x^B\right)^\perp$ has a Schmidt decomposition given by
 \begin{equation}\label{yprime}
 y=\sum_{i}r_i y^{A}_{i}\otimes y^{B}_{i}\;,
 \end{equation}
 where 
 \begin{equation}\label{ort}
 \tr[y^{A}_{i}y^{A*}_{i'}] = \tr[y^{B}_{i}y^{B*}_{i'}]=\delta_{ii'}\;,
 \end{equation}
 and $r_i$ are real non-negative numbers such that $\sum_{i}r_{i}^{2}=1$.
Therefore, substituting~\eqref{yprime} into the right-hand-side of~\eqref{beg} gives
\begin{align*}
\tr\left[y^*\Psi(y)\right] & =\sum_{i,i'}r_{i}r_{i'}\tr\left[y^{A*}_{i'}\otimes y^{B*}_{i'}\Psi(y^{A}_{i}\otimes y^{B}_{i})\right]\\
& =\sum_{i,i'}r_ir_{i'}\left[\sum_{j,\ell}\left(y^{A*}_{i'}\right)_{j\ell}\left(y^{A}_{i}\right)_{\ell j}\left(\phi_{j\ell}^{+A}+\phi_{j\ell}^{-A}\right)\sum_{k,m}\left(y^{B*}_{i'}\right)_{km}\left(y^{B}_{i}\right)_{m k}\left(\phi_{km}^{+B}+\phi_{km}^{-B}\right)\right]\\
&=\sum_{i,i'}r_ir_{i'}\tr\left[y_{i'}^{A*}\left(\Phi_{\rho^A}^{-}+\Phi_{\rho^A}^{+}\right)(y_{i}^{A})\right]
\tr\left[y_{i'}^{B*}\left(\Phi_{\rho^B}^{-}+\Phi_{\rho^B}^{+}\right)(y_{i}^{B})\right].
\end{align*}
Now, denote
\begin{align}
& A_{i'i}\equiv\tr\left[y_{i'}^{A*}\left(\Phi_{\rho^A}^{-}+\Phi_{\rho^A}^{+}\right)(y_{i}^{A})\right],\nonumber\\
& B_{i'i}\equiv\tr\left[y_{i'}^{B*}\left(\Phi_{\rho^B}^{-}+\Phi_{\rho^B}^{+}\right)(y_{i}^{B})\right].
\end{align}
Since the operators $\Phi_{\rho^A}^{-}+\Phi_{\rho^A}^{+}$ and $\Phi_{\rho^B}^{-}+\Phi_{\rho^B}^{+}$ are self-adjoint, the matrices $A$ and $B$ are Hermitian.  Inequality \eqref{assumption} then gives the following upper bounds.
\begin{lemma} \label{lem.A<Tr} The Hermitian matrices $A,B$ above satisfy
\[ 0\leq A < \Tr[(\rho^A)^\alpha] \qquad \text{and} \qquad 0\leq B < \Tr[(\rho^B)^\alpha]. \]
\end{lemma}

\begin{proof} Fix a unit vector $v = [v_i]$; the claim is that $0\leq \langle v,Av\rangle < \Tr[(\rho^A)^\alpha]$ (and similarly for $B$ and $\rho^B$).  Note that
\begin{align*} \langle v,Av\rangle = \sum_{i,i'} \overline{v_{i'}} v_i A_{i'i} &= \sum_{i,i'} \overline{v_{i'}} v_i \tr\left[y_{i'}^{A*}\left(\Phi_{\rho^A}^{-}+\Phi_{\rho^A}^{+}\right)(y_{i}^{A})\right] \\
&= \Tr\left[\sum_{i,i'} (v_{i'}y^A_{i'})^\ast \left(\Phi_{\rho^A}^{-}+\Phi_{\rho^A}^{+}\right)(v_iy_i^A)\right] = \Tr\left[y^{A\ast}\left(\Phi_{\rho^A}^{-}+\Phi_{\rho^A}^{+}\right)(y^A)\right],
\end{align*}
where $y^A=\sum_i v_iy_i$.  From \eqref{ort}, it follows that $y^A$ is also a unit vector; it is in the subspace $(x^A)^\perp$ (being a linear combination of $y_i$ that are in this subspace). Note that the RHS of the equation above is non-negative since $(p_{j}^{\alpha}-p_{k}^{\alpha})/(p_j-p_k)\geq 0$ for all $j,k$. The rest of the statement of the lemma now follows by \eqref{assumption} (the case for $B$ and $\rho^B$ is similar). \end{proof}

Now, diagonalize $A$ and $B$: $A=U^*D^AU$ and $B=V^*D^BV$ where $U$ and $V$ are unitary matrices and $D^A$ and $D^B$ are diagonal matrices.  Lemma \ref{lem.A<Tr} shows that the diagonal entries of $D^A$ are all less than $\Tr[(\rho^A)^\alpha]$, and the diagonal entries of $D^B$ are all less than $\Tr[(\rho^B)^\alpha]$. With this in mind, we get
\begin{align*}
\tr\left[y^*\Psi(y)\right] =\sum_{i,i'}r_{i}r_{i'}A_{i'i}B_{i'i}
&=\sum_{i,i'}r_{i}r_{i'}\sum_{k}U^{*}_{i'k}D^{A}_{kk}U_{ki}\sum_jV^{*}_{i'j}D^{B}_{jj}V_{ji} \\
&= \sum_{j,k} D^A_{kk}D^B_{jj} \sum_{i,i'} r_ir_{i'} U^\ast_{i'k}U_{ki}V^\ast_{i'j}V_{ji}.
 \end{align*}
Note that the internal sum can be written as
\[ \sum_{i,i'} r_ir_{i'} U^\ast_{i'k}U_{ki}V^\ast_{i'j}V_{ji} = \sum_i r_i U_{ki}V_{ji} \, \sum_{i'} r_{i'} U_{i'k}^\ast V_{i'j}^\ast = \left|\sum_{i}r_iU_{ki}V_{ji}\right|^2.  \]
In particular, the internal sum is positive.  Hence, using Lemma \ref{lem.A<Tr}, we have
\be \label{e.penult} \Tr[y^\ast\Psi(y)] =\sum_{k,j}D^{A}_{kk}D^{B}_{jj}\left|\sum_{i}r_iU_{ki}V_{ji}\right|^2 < \Tr[(\rho^A)^\alpha]\Tr[(\rho^B)^\alpha]\sum_{j,k}\left|\sum_{i}r_iU_{ki}V_{ji}\right|^2. \ee
Finally, we may again expand this sum.  Using the unitarity of $U$ and $V$, we find
\[ \sum_{j,k}\left|\sum_{i}r_iU_{ki}V_{ji}\right|^2 = \sum_{i,i',j,k}r_{i}r_{i'}U^{*}_{i'k}U_{ki}V^{*}_{i'j}V_{ji} = \sum_{i,i'}r_{i}r_{i'}\delta_{i'i}\delta_{i'i} = \sum_i r_i^2 = 1.  \]
Combining this with \eqref{e.penult} yields $\Tr[y^\ast\Psi(y)] < \Tr[(\rho^A)^\alpha]\Tr[(\rho^B)^\alpha]$.  This, combined with \eqref{beg}, proves \eqref{main}, concluding the proof of part~(a) of the theorem. 

To prove part (b), we assume w.l.o.g. that $x^A$ is a strict (i.e. non-degenerate) local maximum of $Q_\alpha$, and $x^B$ is a degenerate local maximum of $Q_\alpha$. We therefore have
\begin{equation} \begin{aligned}
& D_{y^{A}}^{2}Q_\alpha(x^A)<0\;\;\;\;\;\;\forall\;\;y^A\in(x^A)^\perp\\
& D_{y^{B}}^{2}Q_\alpha(x^B)\leq 0\;\;\;\;\;\;\forall\;\;y^B\in(x^B)^\perp.
\end{aligned}\end{equation}

First note that from~\eqref{simple} it follows that $x^A\otimes x^B$ is a \emph{degenerate} local maximum of $Q_\alpha$ (recall $\alpha>1$) in any direction of the form $y=x^A\otimes y^{B}$, with $y^B\in(x^B)^\perp$.
Similar arguments shows that $x^A\otimes x^B$ is a \emph{non-degenerate} local maximum of $Q_\alpha$ in any direction of the form $y=y^A\otimes x^B$ where $y^A\in(x^A)^\perp$. It is therefore left to consider directions $y\in \left(x^A\right)^{\perp}\otimes \left(x^B\right)^\perp$. In these directions $x^A\otimes x^B$ is a \emph{non-degenerate} local maximum of $Q_\alpha$ since $D_{y}^{2}Q_\alpha(x^A\otimes x^B)<0$ for $y\in \left(x^A\right)^{\perp}\otimes \left(x^B\right)^\perp$. To see this, note that~\eqref{e.penult} still holds with strict inequality since the first equation of~\eqref{assumption} still holds, while the second equation of~\eqref{assumption} holds with $\leq$ sign. This concludes the proof of part (b) of the theorem.

\section{Proof lemma~\ref{derivatives}}\label{sec:der}

\subsection{Linearization}

Let $x,y\in\mathcal{K}\subset\C^{n\times n}$. 
We define two Hermitian matrices $X$ and $Y$ in $\mathbb{C}^{2n\times 2n}$ 
corresponding to $x$ and $y$, respectively:
\begin{equation}\label{lin}
 X=\left[\begin{array}{cc} 0 & x\\ x^{*} &0 \end{array}\right]\;\text{and}\;
Y=\left[\begin{array}{cc} 0 & y \\y^{*}&0 \end{array}\right]\;, 
 \end{equation}
 where $0$ stands for the $n\times n$ zero matrix. Note that  $\text{Tr}(X^2)=\text{Tr}(Y^2)=2$ if $\text{Tr}(xx^*)=\text{Tr}(yy^*)=1$. 

 We will also denote by $\mathcal{L}\subset\mathbb{C}^{2n\times 2n}$ 
 the linearization space corresponding to $\mathcal{K}$; that is,
 $$
 \mathcal{L}:=\{Y\in\mathbb{C}^{2n\times 2n}\;|\;y\in\mathcal{K}\}\;,
 $$
 where $Y$ corresponds to $y$ as in Eq.~(\ref{lin}).
 Note that for $Y_1,Y_2\in\mathcal{L}$ and for $r_1,r_2\in\mathbb{R}$, the matrix $r_1Y_1+r_2Y_2$ is also in
 $\mathcal{L}$; that is, $\mathcal{L}$ is a vector space over the \emph{real} numbers.
 
 Define 
 \be
 G_\alpha(X)=\text{Tr}\left[X^{2\alpha}\right]\;. 
 \ee
 Then a key observation is that 
 \be 
 G_\alpha(X)=2Q_\alpha(x)\;.
 \ee 
 Thus, in order to calculate the directional derivatives of $Q_\alpha$ in $\mathcal{K}$ , we will first find the derivatives of $G_\alpha$ in $\cL$, and then translate these calculations back to the space $\mathcal{K}$.
 
 Consider the point $(x+ty)/\sqrt{1+t^2}$ (here $0<t\in\mathbb{R}$), in a neighbourhood of $x$.
 Recall from our discussion earlier that $\tr(xy^*)=0$ and assume the normalization $\text{Tr}(xx^*)=\text{Tr}(yy^*)=1$. This point is mapped to $(X+tY)/\sqrt{1+t^2}$ in the space $\mathcal{L}$.
 Moreover, the condition $\tr(xy^*)=0$ is equivalent to the conditions $\tr(XY)=\tr(X\tilde{Y})=0$, where
 \begin{equation}
\tilde{Y}:=\left[\begin{array}{cc} 0 & iy \\-iy^{*}& 0\end{array}\right]\;. 
 \end{equation}
Now, note that
 \begin{align}\label{approx}
 G_\alpha\left(\frac{X+tY}{\sqrt{1+t^2}}\right)=\frac{1}{(1+t^2)^\alpha}G_\alpha(X+tY)&=G_\alpha(X+tY)-\alpha t^2 G_\alpha(X+tY)+O(t^4).
\end{align}
Recall that $x=\text{diag}\{p_1,...,p_n\}$. 
Hence, the eigenvalues of the matrix $X$ are $\{\lambda_j\}_{j=1,...,2n}$,
where $\lambda_j=\sqrt{p_j}$ for $j=1,...,n$, $\lambda_j=-\sqrt{p_{j-n}}$ for $j=n+1,...,2n$.
We will be working with a basis in which $X=\text{diag}\{\lambda_1,...,\lambda_{2n}\}$. In this basis $Y$ does not have the form
given in Eq.(\ref{lin}). To understand the form of $Y$, we now discuss the diagonalization of $X$.

Recall that the matrix $x\in\mathbb{C}^{n\times n}$ is diagonal 
(the singular value decomposition theorem). 
In order to make $X$ diagonal we conjugate it with the following generalization of the Hadamard matrix:
$$
\text{diag}\{\lambda_1,...,\lambda_{2n}\}=UXU\;\text{and}\;U=\frac{1}{\sqrt{2}}\left[\begin{array}{cc} I & I\\ I &-I
\end{array}\right]
$$
where $I$ is the $n\times n$ identity matrix.
Under this change of basis, $Y$ takes the form
\begin{equation}
Y=\left[\begin{array}{cc} w & -iz\\ iz &-w\end{array}\right]\;,
\label{forms}
\end{equation}
where $w$ and $z$ are the $n\times n$  Hermitian matrices defined by $y=w+iz$, or equivalently:
$$
w\equiv \frac{1}{2}\left(y+y^{*}\right)\;\text{and}\;z=\frac{1}{2i}(y-y^*)\;.
$$

\subsection{Taylor Expansion}

In~\cite{GF} the following analytic matrix Taylor expansion was given.
\begin{theorem}[\cite{GF}, Theorem 2] \label{taylor} Let $A,B\in\mathbb{C}^{n\times n}$, and suppose $A=\mathrm{diag}(\lambda_1,\ldots,\lambda_n)$.  If $f\colon\R\to\C$ is real analytic on an open neighborhood of the eigenvalues $\{\lambda_1,\ldots,\lambda_n\}$ of $A$, then
 \begin{align}\label{polcase1}
 \tr [f(A+tB)]=\tr [f(A)]+
 t\sum_{j=1}^{n}f'(\lambda_j)[B]_{jj}+t^2\sum_{j,k=1}^n \frac{f'(\lambda_j)-f'(\lambda_k)}{2(\lambda_j-\lambda_k)}[B]_{jk}[B]_{kj}+O(t^3)\;.
 \end{align}
 \end{theorem}
\begin{remark}
The expansion above can be naturally generalized to higher than the second order, but second order suffices for our present purposes.  The second order coefficient involves a mixture of a genuine derivative and a difference quotient of $f$.  The expression must be interpreted by taking limits when there is division by $0$; for example, when $\lambda_j=\lambda_k$, the coefficient is 
$$
\frac{f'(\lambda_j)-f'(\lambda_k)}{2(\lambda_j-\lambda_k)}\equiv\frac{1}{2}f''(\lambda_j)\;\;\;\text{for}\;\;\lambda_j=\lambda_k
$$
agreeing with the normal Taylor expansion when $n=1$.
\end{remark}

We would like to extend this Taylor expansion beyond the analytic category.  To give meaning to the expression $f(A+tB)$, we henceforth assume $A,B$ are Hermitian, and then interpret this through functional (spectral) calculus.  We would like to apply \eqref{polcase1} to the function $f(x) = g_\alpha(x) \equiv x^{2\alpha}$, which is $C^2$ (for $\alpha>1$) but not analytic; we will see that a version of \eqref{polcase1} holds, but the error term will not generally be $O(t^3)$ but instead $O(t^{2\alpha})$ for $1<\alpha<3/2$.

Our approach is to study the function $F(t)= \Tr[f(A+tB)]$ as an ordinary calculus function, and apply the single-variable version of Taylor's theorem to it.  We use the Kato--Rellich theorem (cf.\ \cite[p.\ 122]{Kato} see also~\cite{Friedland}) on eigenvalue perturbation, which says the following.

\begin{theorem}[Kato, Rellich] \label{t.Rellich} Let $(a,b)\subset\R$ be an open interval, and let $M\colon(a,b)\to\mathbb{C}^{n\times n}$ be a Hermitian matrix-valued real analytic function.  Then there is a real analytic function $U\colon(a,b)\to U(n,\C)$ (the $n\times n$ unitary matrices) and $n$ real analytic functions $\mu_1,\ldots,\mu_n\colon(a,b)\to\R$ such that
\[ M(t) = U(t)^\ast \mathrm{diag}(\mu_1(t),\ldots,\mu_n(t))U(t). \]
\end{theorem}
Here, we call a matrix-valued function real analytic if all of its entries are standard $\C$-valued real analytic functions.  The Kato--Rellich theorem thus states the remarkable fact that the eigenvalues and eigenvectors of a Hermitian matrix $M(t)$ all depend analytically on the parameter $t$, provided the entries of $M(t)$ depend analytically on $t$.

Specializing to the case that $M(t) = A+tB$ where $A=\mathrm{diag}(\lambda_1,\ldots,\lambda_n)$, the Taylor series of the eigenvalue $\mu_j(t)$ was computed in \cite[Appendix A]{GF}; the result is
\begin{equation} \label{e.mu.Rellich} \mu_j(t) = \lambda_j + t[B]_{jj} + t^2\sum_{k\colon\lambda_k\ne\lambda_j}\frac{1}{\lambda_j-\lambda_k}|[B]_{jk}|^2 + O(t^3). \end{equation}
With this in hand, it is now straightforward to generalize Theorem \ref{taylor} to the present case, where $f$ is generally only a little smoother than $C^2$.

\begin{theorem} \label{t.Taylor.2} Let $(a,b)$ be an open interval in $\R$, let $0<\epsilon<1$, and let $f\in C^{2+\epsilon}(a,b)$, meaning that $f\in C^2(a,b)$ and $f''$ is $\epsilon$-H\"older continuous on $(a,b)$.  Let $A,B\in\C^{n\times n}$ be Hermitian matrices, with $A = \mathrm{diag}(\lambda_1,\ldots,\lambda_n)$.  Then the function $F(t) = \Tr[f(A+tB)]$ is also in $C^{2+\epsilon}(a,b)$, and
\begin{equation} \label{polcase2}
F(t) = \tr [f(A)]+ t\sum_{j=1}^{n}f'(\lambda_j)[B]_{jj}+t^2\sum_{j,k=1}^n \frac{f'(\lambda_j)-f'(\lambda_k)}{2(\lambda_j-\lambda_k)}|[B]_{jk}|^2+O(t^{2+\epsilon}).
\end{equation}
\end{theorem}

\begin{proof} By Theorem \ref{t.Rellich}, the matrix $A+tB$ can be diagonalized as $A+tB = U(t)^\ast\mathrm{diag}(\mu_1(t),\ldots,\mu_n(t))U(t)$ where $\mu_1,\ldots,\mu_n$ are analytic in $t$.  By functional (spectral) calculus, $f(A+tB)$ is defined to be
\[ f(A+tB) = U(t)^\ast\mathrm{diag}(f(\mu_1(t)),\ldots,f(\mu_n(t)))U(t) \]
and so
\[ F(t) = \Tr[f(A+tB)] = \sum_{j=1}^n f(\mu_j(t)). \]
Since the functions $\mu_j$ are analytic and $f\in C^{2+\epsilon}(a,b)$, it follows that the functions $f\circ\mu_j$ are in $C^{2+\epsilon}(a,b)$, and therefore so is their sum $F$.  Therefore, from Taylor's theorem
\[ F(t) = F(0) + tF'(0)+\frac12t^2F''(0) + O(t^{2+\epsilon}). \]
Indeed, to see this, use Taylor's theorem to first order with mean value remainder: $F(t) = F(0) + tF'(0) + \frac12t^2 F''(\xi)$ for some $\xi\in[0,t]$; but since $F''$ is $\epsilon$-H\"older continuous, $|F''(\xi)-F''(0)| \le C|\xi-0|^\epsilon \le Ct^\epsilon$ for some constant $C$.  So $\frac12t^2F''(\xi) = \frac12t^2F''(0) + O(t^{2+\epsilon})$ as required.

It remains only to compute the coefficients $F(0)$, $F'(0)$, and $F''(0)$, which we can now easily do using \eqref{e.mu.Rellich}. First we have
\[ F(0) = \sum_{j=1}^n f(\mu_j(0)) = \sum_{j=1}^n f(\lambda_j) = \Tr[f(A)]. \]
Next, applying the chain rule,
\[ F'(0) = \sum_{j=1}^n f'(\mu_j(0))\mu_j'(0) = \sum_{j=1}^n f'(\lambda_j)[B]_{jj}. \]
Differentiating one more time, we have $\frac{d^2}{dt^2}f(\mu_j(t)) = f''(\mu_j(t))(\mu_j'(t))^2 + f'(\mu_j(t))\mu_j''(t)$.  Setting $t=0$ and summing gives
\[ F''(0) = \sum_{j=1}^n \left[f''(\lambda_j)([B]_{jj})^2 + 2f'(\lambda_j)\sum_{k\colon j\ne k}\frac{1}{\lambda_j-\lambda_k}|[B]_{jk}|^2\right]. \]
To see this has the form given in \eqref{polcase2}, denote by $\delta(j,k) = 0$ if $\lambda_j=\lambda_k$, while $\delta(j,k)=1$ if $\lambda_j\ne\lambda_k$.  Then
\begin{align*} F''(0) = \sum_{j=1}^n f''(\lambda_j)([B]_{jj})^2 + 2\sum_{j,k=1}^n f'(\lambda_j)\frac{\delta(j,k)}{\lambda_j-\lambda_k}|[B]_{jk}|^2. \end{align*}
Break up the second sum into its two copies, and in the second one reverse the roles of $j$ and $k$:
\begin{align*} 2\sum_{j,k=1}^n f'(\lambda_j)\frac{\delta(j,k)}{\lambda_j-\lambda_k}|[B]_{jk}|^2 &= \sum_{j,k=1}^n f'(\lambda_j)\frac{\delta(j,k)}{\lambda_j-\lambda_k}|[B]_{jk}|^2 + \sum_{j,k=1}^n f'(\lambda_k)\frac{\delta(j,k)}{\lambda_k-\lambda_j}|[B]_{kj}|^2 \\
&= \sum_{j,k=1}^n \delta(j,k) \frac{f'(\lambda_j)-f'(\lambda_k)}{\lambda_j-\lambda_k}|[B]_{jk}|^2
\end{align*}
where in the second line we have used the fact that $B$ is Hermitian, so $|[B]_{jk}| = |[B]_{kj}|$.  Now, as noted above, if $\lambda_j=\lambda_k$, we interpret the difference quotient to mean $\frac{f'(\lambda_j)-f'(\lambda_k)}{\lambda_j-\lambda_k} \equiv f''(\lambda_j)$.  As such, if there is some $j\ne k$ with $\lambda_j=\lambda_k$, then this term will appear twice in the sum (once at the index $(j,k)$, and again at the index $(k,j)$), with opposite signs; hence, $\delta(j,k)$ is automatically accounted for whenever $j\ne k$, and so we have
\[ F''(0) = \sum_{j=1}^n f''(\lambda_j)([B]_{jj})^2 + \sum_{1\le j\ne k\le n} \frac{f'(\lambda_j)-f'(\lambda_k)}{\lambda_j-\lambda_k}|[B]_{jk}|^2. \]
The first sum is precisely the missing diagonal terms from the second sum, owing to the fact that $B$ is Hermitian and so $([B]_{jj})^2 = |[B]_{jj}|^2$.  This completes the proof.  \end{proof}

\begin{remark} It is worth noting that an alternate proof is possible, avoiding the Kato-Rellich theorem, using Fourier analysis and the useful identity that, for any $C^1$ function $h$, the difference quotient can be written in the form
\[ \frac{h(x)-h(y)}{x-y} = \int_0^1 h'(sx+(1-s)y)\,ds. \]
This allows one to quickly recover the second order Taylor expansion of $\Tr[f(X+tY)]$ given in \eqref{polcase2} for functions $f$ that are slightly smoother than $C^2$ (functions $f$ that are Fourier transforms of complex measures with finite absolute second moment).  This is an unnecessary technical restriction, but suffices to deal with the desired function $f(x) = x^{2\alpha}$ after a smooth cutoff, and it gives a little more motivation to explain {\em why} the mixed difference quotient derivative appears in the Taylor expansion.  For more details, see \cite[Section 3]{KNPS}.
\end{remark}

We now use Theorem \ref{t.Taylor.2} to compute the Taylor expansion of the function $G_\alpha(X+tY)$ in a neighbourhood of $t=0$, and use it to prove Lemma~\ref{derivatives}.

\subsection{The first and second derivatives of $Q_\alpha(x)$}

We summarize the statement of Lemma~\ref{derivatives} in the following theorem, where we also include the first order derivative. 

\begin{theorem}
Let $\rho\equiv x^2=\diag\{p_1,\ldots,p_{n}\}$. Decompose $y=w+iz$ with $w,z$ Hermitian. Then,
\begin{align}
&D_{y}^{1}Q_\alpha(x) =\alpha\tr(w x^{2\alpha-1})\nonumber\\
&D_{y}^{2}Q_\alpha(x)=2\alpha\left[-\tr\left(\rho^{\alpha}\right)+\tr\left(w\Phi_{\rho}^{-}(w)+z\Phi_{\rho}^{+}(z)\right)
\right]
\label{se}
\end{align}
where $\Phi^\pm_\rho$ are defined in \eqref{e.phi.simple}.

\end{theorem}

\begin{remark}
The condition for $x\in\mathcal{K}$ to be critical is $D_{y}^{1}Q_\alpha(x)=0$ which is equivalent to $\tr[(y^*+y)x^{2\alpha-1}]=0$ for all $y\in\mathcal{K}$ such that $\tr(xy^*)=0$. Moreover, if $x$ is critical then we also have 
$D_{iy}^{1}Q_\alpha(x)=0$ for all $y\in x^\perp\subset\mathcal{K}$.
Hence, if $x$ is critical we must have 
\begin{equation}\label{critical}
\tr(y^*x^{2\alpha-1})=\tr(yx^{2\alpha-1})=0
\end{equation} 
for all $y\in x^\perp\subset\mathcal{K}$.   
\end{remark}
\begin{proof} Since $\alpha>1$, the function $g_\alpha(x) = x^{2\alpha}$ is $C^{2+\epsilon}(\R)$ for any $\epsilon\le 2(\alpha-1)$.  (We think of $g_\alpha(x) = (x^2)^\alpha$, so $g'(x) = 2\alpha (x^2)^{\alpha-1/2}$ and $g''(x) = 2\alpha(2\alpha-1)(x^2)^{\alpha-1}$.)  Therefore, we may apply Theorem \ref{t.Taylor.2} to the function $f=g_\alpha$, with $A=X$ and $B=Y$.  Applying the expansion~\eqref{polcase2} to $F(t) = G_\alpha(X+tY) = \Tr[g_\alpha(X+tY)]$ gives
\begin{align}\label{tay}
 G_\alpha(X+tY)=G_\alpha(X)+
 t\sum_{j=1}^{2n}g'_\alpha(\lambda_j)Y_{jj}+t^2\sum_{j,k=1}^{2n} \frac{g'_\alpha(\lambda_j)-g'_\alpha(\lambda_k)}{2(\lambda_j-\lambda_k)}\left|Y_{jk}\right|^2+O(t^{2+\varepsilon}).
 \end{align}
 
Recalling that $\lambda_j=\sqrt{p_j}$ for $j=1,...,n$, $\lambda_j=-\sqrt{p_{j-n}}$ for $j=n+1,...,2n$ gives
\be
\sum_{j=1}^{2n}g'_\alpha(\lambda_j)Y_{jj}=2\alpha\sum_{j=1}^{n}p^{\alpha-1/2}_{j}w_{jj}=2\alpha\tr\left(x^{2\alpha-1}w\right)
\ee
where we have used the form~\eqref{forms} of $Y$.
Similarly, for the second terms in~\eqref{tay} we get
\begin{align}
\sum_{j,k=1}^{2n} \frac{g'_\alpha(\lambda_j)-g'_\alpha(\lambda_k)}{2(\lambda_j-\lambda_k)}\left|Y_{jk}\right|^2
=
\sum_{j,k=1}^{n} \frac{g'_\alpha(\sqrt{p_j})-g'_\alpha(\sqrt{p_k})}{\sqrt{p_j}-\sqrt{p_k}}\left|w_{jk}\right|^2+\sum_{j,k=1}^{n} \frac{g'_\alpha(\sqrt{p_j})+g'_\alpha(\sqrt{p_k})}{\sqrt{p_j}+\sqrt{p_k}}\left|z_{jk}\right|^2
\end{align}
where we have used $g_\alpha(-\sqrt{p_k})=-g_\alpha(\sqrt{p_k})$.
Substituting $g'_\alpha(\sqrt{p_j})=2\alpha p_{j}^{\alpha-1/2}$
gives
\be
\sum_{j,k=1}^{2n} \frac{g'_\alpha(\lambda_j)-g'_\alpha(\lambda_k)}{2(\lambda_j-\lambda_k)}\left|Y_{jk}\right|^2
=2\alpha\tr\left(w\Phi_{\rho}^{-}(w)+z\Phi_{\rho}^{+}(z)\right).
\ee
Now, combining \eqref{approx} with \eqref{polcase2} to $0$th order, we have

\begin{align*} G_\alpha\left(\frac{X+tY}{\sqrt{1+t^2}}\right)=\frac{1}{(1+t^2)^\alpha}G_\alpha(X+tY)&=G_\alpha(X+tY)-\alpha t^2 G_\alpha(X+tY)+O(t^4) \\
&= G_\alpha(X+tY) - \alpha t^2 G_\alpha(X) + O(t^3)
\end{align*}
and combining this with the full force of \eqref{polcase2} yields
\begin{equation}
 G_\alpha\left(\frac{X+tY}{\sqrt{1+t^2}}\right)=G_\alpha(X)+t2\alpha\tr\left(x^{2\alpha-1}w\right)+\alpha t^2\left[-G_\alpha(X)
 +2\tr\left(w\Phi_{\rho}^{-}(w)+z\Phi_{\rho}^{+}(z)\right)\right]+O(t^{2+\varepsilon}).
 \end{equation}
Finally, since $Q_\alpha(x)=\frac12 G_\alpha(X)$, we conclude
 \be
 Q_\alpha\left(\frac{X+tY}{\sqrt{1+t^2}}\right)=Q_\alpha(x)+t\alpha\tr\left(x^{2\alpha-1}w\right)+\alpha t^2\left[-Q_\alpha(x)
 +\tr\left(w\Phi_{\rho}^{-}(w)+z\Phi_{\rho}^{+}(z)\right)\right]+O(t^{2+\varepsilon}).
 \ee
 This completes the proof.  \end{proof}

\section{Proof of Lemma~\ref{directions}}\label{sec:dir}
Consider the general form of $y$ given in~\eqref{generalform}. Writing $y=w+iz$ with $w,z$ Hermitian, we therefore have
\begin{align}
w & =c_1x^A\otimes w^B+c_2w^A\otimes x^B+c_3w'\nonumber\\
z & =c_1x^A\otimes z^B+c_2\omega^A\otimes x^B+c_3z'
\end{align}
where we decomposed $y^{A}=w^A+iz^A$, $y^B=w^B+iz^B$, and $y'=w'+iz'$.
Now substituting these into the formula~\eqref{deriv} of Lemma~\ref{derivatives} for the second directional derivative of $Q_\alpha$ gives
\begin{align}\label{crossterms}
\frac{1}{2\alpha}D_{y}^{2}Q_\alpha(x^A\otimes x^B)  =-\tr\left[\left(\rho^A\otimes\rho^B\right)^\alpha\right]
&+c_{1}^{2}\tr\left[x^A\otimes w^B\Phi_{\rho^A\otimes\rho^B}^{-}(x^A\otimes w^B)+x^A\otimes z^B\Phi_{\rho^A\otimes\rho^B}^{+}(x^A\otimes z^B)\right]\nn
&+c_{2}^{2}\tr\left[w^A\otimes x^B\Phi_{\rho^A\otimes\rho^B}^{-}(w^A\otimes x^B)+z^A\otimes x^B\Phi_{\rho^A\otimes\rho^B}^{+}(z^A\otimes x^B)\right]\nn
&+c_{3}^{2}\tr\left[w'\Phi_{\rho^A\otimes\rho^B}^{-}(w')+z'\Phi_{\rho^A\otimes\rho^B}^{+}(z')\right]+\text{cross terms}\nn
& =c_{1}^{2}D_{x^A\otimes y^{B}}^{2}Q_\alpha(x)+c_{2}^{2}D_{y^{A}\otimes x^B}^{2}Q_\alpha(x)+c_{3}^{2}D_{y'}^{2}Q_\alpha(x)+\text{cross terms}
\end{align} 
where in the last equality we have used the normalization $c_{1}^{2}+c_{2}^{2}+c_{3}^{2}=1$. The cross terms are all the elements that have two distinct terms to the right and left of $\Phi_{\rho^A\otimes\rho^B}^{\pm}$. We now show that all these terms are zero if $x^A$ and $x^B$ are critical points.

First, recall that w.l.o.g.\ we assume that both $x^A$ and $x^B$ are square diagonal matrices. Moreover, since they are critical points we get from~\eqref{critical} that
\be\label{cri}
\tr\left[w^A(x^A)^{2\alpha-1}\right]=\tr\left[z^A(x^A)^{2\alpha-1}\right]=0
\ee
and
\be
\tr\left[w^B(x^B)^{2\alpha-1}\right]=\tr\left[z^B(x^B)^{2\alpha-1}\right]=0\;.
\ee

Next, note that
\be
\left[\Phi_{\rho^A\otimes\rho^B}^{\pm}(y)\right]_{jk,\ell m}=\phi^{AB\pm}_{jk,\ell m}y_{jk,\ell m}\;\;\text{ where }\;\;
\phi^{AB\pm}_{jk,\ell m}\equiv\frac{(p_jq_k)^{\alpha-1/2}\pm (p_\ell q_m)^{\alpha-1/2}}{(p_jq_k)^{1/2}\pm (p_\ell q_m)^{1/2}}
\ee
where $\{p_j\}$ and $\{q_k\}$ are the eigenvalues of the diagonal matrices $\rho^A\equiv\left(x^A\right)^2$ and $\rho^B=(x^B)^2$. Note also that $\psi^{AB\pm}_{jk,\ell k}=\phi^{A\pm}_{j\ell}q^{\alpha-1}_{k}$ and similarly $\phi^{AB\pm}_{jk,j m}=p_{j}^{\alpha-1}\phi^{B\pm}_{km}$. Therefore, since $x^A$ is diagonal we get
\be
\Phi_{\rho^A\otimes\rho^B}^{\pm}(x^A\otimes w^B)=\left(x^A\right)^{2\alpha-1}\otimes\Phi_{\rho^B}^{\pm}(w^B),
\ee
and similarly
\be
\Phi_{\rho^A\otimes\rho^B}^{\pm}(w^A\otimes x^B)=\Phi_{\rho^A}^{\pm}(w^A)\otimes\left(x^B\right)^{2\alpha-1}\;.
\ee
Therefore, computing the first cross terms,
\begin{align}
& \tr\left[x^A\otimes w^B\Phi_{\rho^A\otimes\rho^B}^{-}(w^A\otimes x^B)\right]=\tr\left[w^A\otimes x^B\Phi_{\rho^A\otimes\rho^B}^{-}(x^A\otimes w^B)\right]\nn
&=\tr\left[\left(w^A\otimes x^B\right)\left(\left(x^A\right)^{2\alpha-1}\otimes\Phi_{\rho^B}^{-}(w^B)\right)\right]
=\tr\left[w^A\left(x^A\right)^{2\alpha-1}\right]\tr\left[ x^B\Phi_{\rho^B}^{-}(w^B)\right]=0\;,
\end{align}
where the first equality follows from the fact that $\Phi_{\rho^A\otimes\rho^B}^{-}$ is self-adjoint, and the last one from~\eqref{cri}.
Next, the cross terms
\be\label{cross}
\tr\left[x^A\otimes w^B\Phi_{\rho^A\otimes\rho^B}^{-}(w')\right]=\tr\left[w'\Phi_{\rho^A\otimes\rho^B}^{-}(x^A\otimes w^B)\right]=\tr\left[w'\left(\left(x^A\right)^{2\alpha-1}\otimes\Phi_{\rho^B}^{-}(w^B)\right)\right]\;.
\ee
To see that this term is also zero, recall the expression for $y'$ in~\eqref{yprime}. It gives
\be
w'=\frac{y'+y^{\prime*}}{2}=\frac{1}{2}\sum_{i}r_i \left(y^{A}_{i}\otimes y^{B}_{i}+y^{*A}_{i}\otimes y^{*B}_{i}\right)\;.
\ee
Hence, the right hand side of~\eqref{cross} becomes
\begin{align}
&\tr\left[w'\left(\left(x^A\right)^{2\alpha-1}\otimes\Phi_{\rho^B}^{-}(w^B)\right)\right] =\frac{1}{2}\sum_{i}r_i \tr\left[\left(y^{A}_{i}\otimes y^{B}_{i}+y^{*A}_{i}\otimes y^{*B}_{i}\right)\left(\left(x^A\right)^{2\alpha-1}\otimes\Phi_{\rho^B}^{-}(w^B)\right)\right]\nn
&=
\frac{1}{2}\sum_{i}r_i\left(\tr\left[y^{A}_{i}\left(x^A\right)^{2\alpha-1}\right]\tr\left[y^{B}_{i}\Phi_{\rho^B}^{-}(w^B)\right]+
\tr\left[y^{*A}_{i}\left(x^A\right)^{2\alpha-1}\right]\tr\left[y^{*B}_{i}\Phi_{\rho^B}^{-}(w^B)\right]\right)=0\;,
\end{align}
where in the last equality we have used~\eqref{critical} for $y_{i}^{A}$ and $y_{i}^{*A}$. Using similar arguments for the final cross terms yields
\be
\tr\left[w^A\otimes x^B\Phi_{\rho^A\otimes\rho^B}^{-}(w')\right]=\tr\left[w'\Phi_{\rho^A\otimes\rho^B}^{-}(w^A\otimes x^B)\right]=0\;.
\ee
Therefore, we have shown that all the cross terms of $w$ in~\eqref{crossterms} are zero.
Using similar arguments it follows that all the cross terms of $z$ are also zero. This completes the proof of Lemma~\ref{directions}.

\section{Proof of Lemma~\ref{additivity}}\label{sec:add}
We will prove that 
$\Phi_{\rho^{A}\otimes\rho^{B}}^{\pm}\leq \Psi$ by showing it for the components. That is, we will show that 
$\phi_{jk,\ell m}^{\pm}\leq \psi_{jk,\ell m}$ for all indices $j,k,\ell,m$; it is straightforward to verify that this is equivalent to the operator inequality for Haadamard product operators. This componentwise inequality is equivalent to 
\be \label{a} \frac{(p_jq_k)^{\alpha-1/2}\pm (p_\ell q_m)^{\alpha-1/2}}{(p_jq_k)^{1/2}\pm (p_\ell q_m)^{1/2}} \le \frac{p_j^\alpha- p_\ell^\alpha}{p_j-p_\ell}\frac{q_k^\alpha-q_m^\alpha}{q_k-q_m}. \ee

The following simple lemma shows that $\Phi_{\rho^{A}\otimes\rho^{B}}^{\pm}\leq \Psi$ if and only if
$\Phi_{\rho^{A}\otimes\rho^{B}}^{-}\leq \Psi$.

\begin{lemma} \label{lemma.diff.quotient} Let $r,s\ge 0$ and $\beta\ge1$.  Then
\be \label{e.ineq.30} \frac{r^\beta-s^\beta}{r-s} \ge \frac{r^\beta + s^\beta}{r+s}\ge 0. \ee
\end{lemma}
\begin{proof} First suppose $r\ne 0$ and $s\ne r$.  Dividing through by $r$ and setting $t=s/r$, the desired inequalities are
\be \label{e.ineq.31} \frac{1-t^\beta}{1-t} \ge \frac{1+t^\beta}{1+t} \ge 0. \ee
The second inequality is manifestly satisfied.  It is also easy to see that $\frac{1-t^\beta}{1-t}\ge 0$ (in fact whenever $\beta\ge 0$), simply by considering the two cases $t<1$ and $t>1$.  For the first inequality in \eqref{e.ineq.31}, we simplify
\be \frac{1-t^\beta}{1-t} - \frac{1+t^\beta}{1+t} = \frac{2t}{1+t}\frac{1-t^{\beta-1}}{1-t} \ee
and, by what we just showed, this is $\ge 0$ as well.

Now, if $s=r\ne 0$, we interpret the terms by taking the limit $s\to r$, which corresponds to $t\to 1$, and so \eqref{e.ineq.31} becomes $\beta\ge 1\ge 0$, which is true given the assumptions of the lemma.  Finally, if $r=0$, then \eqref{e.ineq.30} is the true statement $1\ge 1\ge 0$ if $s\ne 0$ (and similarly if $s=0$, evaluated by taking the limit $s\to0$).  This concludes the proof. \end{proof}

Applying the above lemma to \eqref{a}, with $r=\sqrt{p_jq_k}$ and $s=\sqrt{p_\ell q_m}$ shows that it is sufficient to prove:
\be 
\label{b} \frac{(p_jq_k)^{\alpha-1/2}- (p_\ell q_m)^{\alpha-1/2}}{(p_jq_k)^{1/2}- (p_\ell q_m)^{1/2}} \le \frac{p_j^\alpha- p_\ell^\alpha}{p_j-p_\ell}\frac{q_k^\alpha-q_m^\alpha}{q_k-q_m}. 
\ee
We first contend with some degenerate cases.  Suppose either $p_j=0$ or $q_k=0$; then the inequality reduces to
\be \label{e.reduced.ineq} (p_\ell q_m)^{\alpha-1} \le p_\ell^{\alpha-1} \frac{q_k^\alpha-q_m^\alpha}{q_k-q_m}. \ee
If $p_\ell=0$ this holds vacuously as $0\le 0$; otherwise we divide through by $p_\ell^{\alpha-1}$, giving
\be \label{e.reduced.ineq.2} q_m^{\alpha-1} \le \frac{q_k^\alpha-q_m^\alpha}{q_k-q_m}. \ee
It is easy to verify that this holds true for all $q_m\ge0$ and all $\alpha\ge 1$.  Thus, \eqref{a} holds true in these degenerate cases.  We therefore assume $p_j,q_k>0$.  Henceforth, let $s=\frac{p_\ell}{p_j}$ and $t=\frac{q_m}{q_k}$. Dividing both sides of \eqref{a} through by $(p_jq_k)^{\alpha-1/2}$, our final goal is to prove the following.

\begin{proposition} \label{c.strong.conj} For all $s,t\ge 0$ and $\alpha\ge 1$,
\be \label{con} \frac{1-(st)^{\alpha-1/2}}{1-\sqrt{st}}\leq
\frac{1-s^\alpha}{1-s}\frac{1-t^\alpha}{1-t}. \ee
The inequality is strict if $\alpha>1$ and at least one of $s,t$ is $>0$.
\end{proposition}

\begin{proof} Inequality \eqref{con} takes the form $f(st) \le g(s)g(t)$ where
\be f(t) = \frac{1-t^{\alpha-1/2}}{1-t^{1/2}}, \qquad g(t) = \frac{1-t^\alpha}{1-t}. \ee
First consider the case $s=t$.  We compute
\be \label{e.g-f} g(t)^2 - f(t^2) = \frac{(1-t^\alpha)^2}{(1-t)^2} - \frac{1-t^{2\alpha-1}}{1-t} = \frac{(1-t^\alpha)^2 - (1-t)(1-t^{2\alpha-1})}{(1-t)^2}. \ee
The numerator simplifies to
\be \label{e.s=t} (1-2t^\alpha + t^{2\alpha}) - (1-t-t^{2\alpha-1}+t^{2\alpha}) = -2t^\alpha + t + t^{2\alpha-1} = t(t^{2(\alpha-1)} - 2t^{\alpha-1}+1) = t(t^{\alpha-1}-1)^2 \ge 0. \ee
Thus, we know $f(t^2) \le g(t)^2$.  We would like to conclude that $f(st) \le g(s)g(t)$.  Let $r=\sqrt{st}$; what we just proved shows that
\be \label{e.con2} f(st) = f(r^2) \le g(r)^2 = g(\sqrt{st})^2. \ee
Therefore, to prove the desired inequality, it suffices to show that $g(\sqrt{st})^2 \le g(s)g(t)$.  Now, $g$ is a positive function of a positive variable, so we can define a new function $h(\xi) = \ln g(e^\xi)$, where $\xi\in\R$.  Then the requirement that $g(\sqrt{st})^2 \le g(s)g(t)$ becomes the statement that
\be h\left(\frac12(\xi + \zeta)\right) \le \frac12\left(h(\xi)+h(\zeta)\right). \ee
Therefore, the proof will be complete once we show that $h$ is convex.

To be explicit, the function $h$ is
\be h(\xi) = \ln\frac{1-e^{\alpha \xi}}{1-e^\xi}. \ee
The function is manifestly smooth for $\xi\ne 0$, and is continuous on $\R$ if we define its value at $0$ to be the limit $\ln\alpha$.  Note also that
\be h(-\xi) = (1-\alpha)\xi + h(\xi). \ee
Hence, it suffices to show that $h$ is convex on $(0,\infty)$.  On this domain, $h(\xi) = \ln(e^{\alpha\xi}-1)-\ln(e^\xi-1)$, and so
\be h''(\xi) =  -\frac{\alpha^2 e^{\alpha\xi}}{(e^{\alpha\xi}-1)^2}+\frac{e^\xi}{(e^\xi-1)^2}. \ee
Our goal is to show that $h''(\xi)\ge 0$ for all $\xi>0$.  Note that $h''(\xi) = \upsilon(1,\xi)-\upsilon(\alpha,\xi)$, where
\be \upsilon(\alpha,\xi) = \frac{\alpha^2 e^{\alpha\xi}}{(e^{\alpha\xi}-1)^2}. \ee
Hence, to show the desired conclusion that $h''(\xi)\ge 0$ for all $\xi>0$, it suffices to show that for each $\xi>0$ the function $\alpha\mapsto\upsilon(\alpha,\xi)$ is decreasing.  We compute the derivative
\be \label{e.d/dalpha} \frac{\partial}{\partial\alpha}\upsilon(\alpha,\xi) = -\alpha \frac{e^{\alpha \xi}}{(e^{\alpha\xi}-1)^3}[(\alpha \xi-2)e^{\alpha \xi} + \alpha \xi+2]. \ee
The factor $-\alpha \frac{e^{\alpha \xi}}{(e^{\alpha\xi}-1)^3}$ is $<0$.  The remaining factor takes the form $\chi(\alpha\xi)$, where
\be \chi(u) = (u-2)e^u + u + 2. \ee
Elementary calculus shows that $\chi$ is smooth, $\chi'(u) = (u-1)e^u + 1$, and $\chi''(u) = ue^u$.  In particular, $\chi(0) = \chi'(0) = \chi''(0)=0$.  Since $\chi''(u)>0$ for $u>0$, $\chi'$ is increasing on this domain, so $\chi'(u) > \chi'(0)=0$.  Thus $\chi$ is increasing, and since $\chi(0)=0$, $\chi(u)\ge 0$.  We conclude that $\frac{\partial}{\partial \alpha}\upsilon(\alpha,\xi) < 0$ for $\alpha,\xi>0$, as desired, thus proving \eqref{con}.

As to the strictness: \eqref{e.g-f} and \eqref{e.s=t} show that $f(t^2)<g(t)^2$ for $t>0$ (the case $s=t=1$ reduces \eqref{con} to $2\alpha-1 \le \alpha^2$ which is strict for $\alpha>1$).  Hence, we also get strictness in \eqref{con} whenever $st\ne 0$, thanks to \eqref{e.con2}.  If only one of $s,t$ is $0$ (say $s=0$ but $t\ne 0$), then \eqref{con} becomes $1 \le \frac{1-t^\alpha}{1-t}$, which is easily verified to be strict for $t>0$ and $\alpha>1$.  \end{proof}

\begin{remark} In fact, to conclude that the derivative in \eqref{e.d/dalpha} is $<0$ only needed $\alpha>0$.  But since $h_\alpha''(\xi) = \upsilon(1,\xi)-\upsilon(\alpha,\xi)$, this means that for $0<\alpha<1$, the function $h_\alpha$ is actually {\em concave}.  Hence our proof of \eqref{con} fails in this regime.  (This does not, however, mean that \eqref{con} is necessarily false there.) \end{remark}

\section{Concluding Remarks}\label{sec:concluding}

We have shown that the minimum Renyi entropy output of a quantum channel with Reyni parameter $\alpha>1$
is locally additive. This result extends the work of~\cite{GF} from $\alpha=1$ to $\alpha> 1$, and thereby demonstrates that local additivity holds for a large class of entropy functionals.  In~\cite{HW} Hayden and Winter showed that there are counterexamples for the global additivity conjecture for all Renyi entropies with $\alpha>1$. Hence, the current work complements their result by showing that these counter examples corresponds to a global effect of quantum channels, and cannot be a consequence of local properties of the channels involved.

In Appendix B of~\cite{FGA} (see also~\cite{GF}), it was shown that
both the local and global additivity conjectures are false for all Renyi entropies over the real numbers.
This in turn implies that a straightforward argument involving just directional derivatives
could not provide a proof of local additivity in the general complex case. Hence,
our method to prove local additivity strongly involved the complex structure. In particular, in~\eqref{complex} we use explicitly the assumption that $D_{y}^{2}Q_\alpha(x^A)<0$ in both directions $y=y^A$ and $y=iy^A$, where $i=\sqrt{-1}$.

While both proofs of local additivity for $\alpha=1$ and $\alpha>1$ use explicitly the complex structure, they exhibit key differences. The distinction follows from the fact that in the case $\alpha>1$ we essentially prove local multiplicativity of the $Q_\alpha$ functions, whereas in the von-Neumann case we prove local additivity. The later is somewhat more simple since certain cross-terms cancel out due to the additivity property of the von-Neumann case.  In particular, given $y=\sum_{i}r_iy_{i}^{A}\otimes y_{i}^{B}\in \left(x^A\right)^\perp\otimes \left(x^B\right)^\perp$, we had to show that
$D_{y}^{2}E_\alpha(x^A\otimes x^B)>0$ (or equivalently $D_{y}^{2}Q_\alpha(x^A\otimes x^B)<0$). In the $\alpha>1$ case, this was done using the fact that $D_{y^A}^{2}E_\alpha(x^A)>0$ and $D_{y^B}^{2}E_\alpha(x^B)>0$ for any 
$y^A\in\text{span}\{y_{i}^{A}\}$ and $y^{B}\in\text{span}\{y_{i}^{B}\}$, respectively. On the other hand, in the case $\alpha=1$, all we needed to use is that $D_{y^{A}_{i}}^{2}E_\alpha(x^A)>0$ and $D_{y^{B}_{i}}^{2}E_\alpha(x^B)>0$ for all $i$.
This simplification was possible in the $\alpha=1$ case since the additive nature of the von-Neumann entropy led to the cancellation of the cross terms in the linear combination of $y=\sum_{i}r_iy_{i}^{A}\otimes y_{i}^{B}$. This cancellation does not occur in the $\alpha>1$ case, and instead we had to diagonalize the matrices $A$ and $B$ and use other arguments (see Lemma~5 and the arguments below it).

Another key difference between the $\alpha=1$ case and the $\alpha>1$ case is related to Lemma 4.
Lemma~4 in the limit $\alpha\to 1$ does \emph{not} reduce to the analogous lemma that was used in the $\alpha=1$ case. Again, the main reason for this is the multiplicativity versus additivity properties of the $\alpha>1$ case and the $\alpha=1$ case, respectively. In particular, Lemma~4 (or Proposition~\ref{c.strong.conj} that is used to prove Lemma 4) becomes trivial in the limit $\alpha\to 1$ and cannot be used to prove local additivity for the case $\alpha=1$.

In both cases of $\alpha=1$ (see~\cite{GF}) and $\alpha>1$ we had to assume that at least one of the two local minima is strict. The main reason for that is that otherwise it seems to be possible that  $D_{y}^{2}E_\alpha(x^A\otimes x^B)=0$ (rather than strictly positive) for some $y\in\left(x^A\right)^\perp\otimes \left(x^B\right)^\perp$. In order to study this case,
one will need to study third and fourth order directional derivatives which lead to very cumbersome expressions. It is therefore left open if local additivity holds in this case.

Finally, the case $\alpha<1$ was not studied in this paper since Lemma 4 fails to hold in this limit (in fact, in this case we need a similar lemma with the inequality {\em reversed}, since we are interested in local {\em minima} and not local maxima of $Q_\alpha$). Hence, the techniques used here can not be applied directly to this case, and we leave the study of this case for future work.\\

\subsection*{Acknowledgments}
We extend thanks to Mark Girard for many stimulating discussions on topics that are closely related to this work.

\end{document}